\documentclass[english,11pt]{article}
\usepackage{amstext}
\usepackage{amsmath}
\usepackage{amsthm}
\usepackage{amssymb}
\usepackage{graphicx}
\usepackage{babel}
\usepackage{xcolor}
\usepackage{times}
\usepackage{xspace}
\usepackage{paralist}
\usepackage{wrapfig}
\usepackage[font=footnotesize]{caption}
\usepackage[top=1.in, bottom=1.in, left=1.in, right=1.in]{geometry}

\newcommand{\denselist}{\itemsep 0pt\parsep=1pt\partopsep 0pt}
\newcommand{\bitem}{\begin{itemize}\denselist}
\newcommand{\eitem}{\end{itemize}}
\newcommand{\benum}{\begin{enumerate}\denselist}
\newcommand{\eenum}{\end{enumerate}}

\newcommand{\defn}[1]           {{\textit{\textbf{\boldmath #1\/}}}}
\newcommand{\Prob}[1]{\text{Prob}\left\{ #1 \right\}}

\def\loglog{\operatorname{loglog}}
\def\poly{\operatorname{poly}}
\def\polylog{\operatorname{polylog}}

\DeclareGraphicsExtensions{.pdf,.png,.jpg}
\def\eps{\varepsilon}

\newtheorem{theorem}{Theorem}
\newtheorem{lemma}[theorem]{Lemma}

\newtheorem{definition}{Definition}

\newtheorem{proposition}{Proposition}

\newenvironment{restate}[1]
{\par \vspace{1.5ex} \noindent \textbf{Statement of #1.}  \em}
{\par}

\newcommand{\full}[1]{}               

\newcommand{\old}[1]{{}}

\newlength\abovesectionskip
\newlength\belowsectionskip
\def\sectionfont{\normalfont\Large\bfseries}
\newlength\abovesubsectionskip
\newlength\belowsubsectionskip
\def\subsectionfont{\normalfont\large\bfseries}
\newlength\abovesubsubsectionskip
\newlength\belowsubsubsectionskip
\def\subsubsectionfont{\normalfont\normalsize\bfseries}
\newlength\aboveparagraphskip
\newlength\belowparagraphskip
\def\paragraphfont{\normalfont\normalsize\bfseries}
\makeatletter
\def\section{\@startsection{section}{1}{\z@}{-\abovesectionskip}%
              {\belowsectionskip}{\sectionfont}}
\def\subsection{\@startsection{subsection}{1}{\z@}{-\abovesubsectionskip}%
                 {\belowsubsectionskip}{\subsectionfont}}
\def\subsubsection{\@startsection{subsubsection}{1}{\z@}%
                    {-\abovesubsubsectionskip}{\belowsubsubsectionskip}%
                    {\subsubsectionfont}}
\def\paragraph{\@startsection{paragraph}{4}{\z@}{-\aboveparagraphskip}%
                {\belowparagraphskip}{\paragraphfont}}
\makeatother

\usepackage{times}

\def\compactify{\itemsep=0pt \topsep=0pt \partopsep=0pt \parsep=0pt}
\let\latexusecounter=\usecounter

\begin{document}

\title{Complex Contagions in Kleinberg's Small World Model}
\author{Roozbeh Ebrahimi\thanks{Department of Computer Science, Stony Brook University, Stony Brook, NY 11794. \texttt{\{rebrahimi,jgao,gghasemiesfe\}@cs.stonybrook.edu}} \and Jie Gao$^\ast$ \and Golnaz Ghasemiesfeh$^\ast$
\and Grant Schoenebeck\thanks{Department of Computer Science and Engineering, University of Michigan, Ann Arbor, Michigan, MI 48109. \texttt{schoeneb@umich.edu}}}

\clearpage\maketitle
\thispagestyle{empty}



\begin{abstract}

Complex contagions describe diffusion of behaviors in a social network in settings where spreading requires the influence by two or more neighbors. In a $k$-complex contagion, a cluster of nodes are initially infected, and additional nodes become infected in the next round if they have at least $k$ already infected neighbors. It has been argued that complex contagions better model behavioral changes such as adoption of new beliefs, fashion trends or expensive technology innovations. This has motivated rigorous understanding of spreading of complex contagions in social networks. Despite simple contagions ($k=1$) that spread fast in all small world graphs, how complex contagions spread is much less understood. Previous work~\cite{Ghasemiesfeh:2013:CCW} analyzes complex contagions in Kleinberg's small world model~\cite{kleinberg00small} where edges are randomly added according to a spatial distribution (with exponent $\gamma$) on top of a two dimensional grid structure. It has been shown in~\cite{Ghasemiesfeh:2013:CCW} that the speed of complex contagions  differs exponentially when $\gamma=0$ compared to when $\gamma=2$.

In this paper, we fully characterize the entire parameter space of $\gamma$ except at one point, and provide upper and lower bounds for the speed of $k$-complex contagions.  We study two subtly different variants of Kleinberg's small world model and show that, with respect to complex contagions,  they behave differently.  For each model and each $k \geq 2$, we show that there is an intermediate range of values, such that when $\gamma$ takes any of these values, a $k$-complex contagion spreads quickly on the corresponding graph, in a polylogarithmic number of rounds.  However, if $\gamma$ is outside this range, then  a $k$-complex contagion requires a polynomial number of rounds to spread to the entire network.
\end{abstract}

\newpage
\pagestyle{plain}\setcounter{page}{1}


\section{Introduction}

Social behavior is undoubtedly one of the defining characteristics of us as a species. Social acts are influenced by the behavior of others while at same time influencing them.
New social behaviors may emerge and spread in the network like a contagion. Some of these contagions are beneficial (e.g., adopting healthy lifestyle) or profitable (e.g., viral marketing), while some others are destructive and undesirable (such as teenager smoking, alcohol abuse, or vandalism). To effectively promote desirable contagions and discourage undesirable ones, the first step is to understand how these contagions spread in networks and what are the important parameters that lead to fast spreading.

Social contagions can be categorized by the way they spread in networks. Our focus in this paper is on contagions that are \emph{complex}, contagions that require social reaffirmation from multiple neighbors, as opposed to \emph{simple} ones, which can spread through a single contact. Simple contagions are adequate models for many spreading phenomena such as rumors, disease, etc. But, when a spreading contagion is concerned with individual's actions and behavioral changes, it has been argued in sociology literature that complex contagions represent most of the realistic settings. This model of contagion makes an important distinction between the \emph{acquisition} of information and the decision to \emph{act} on the information. While it takes only a single tie for people to hear about a new belief, technology, fad or fashion, ``it is when they see people they know getting involved, that they become most susceptible to recruitment'', as Centola and Macy~\cite{G08} explain.

Many examples of complex contagions have been reported in social studies, including buying pricey technological innovations, changes in social behaviors, the decision to migrate, etc.~\cite{Coleman:1966,centola2010spread}. Studies of large scale data sets from online social networks, like Facebook and Twitter, have confirmed the existence of the complex contagion phenomenon as well~\cite{ugander12,Romero11}.

The speed of simple contagions is inherently linked to the diameter of the network. As such, almost all generative social network models support fast (polylogarithmic) spreading of simple contagions, because they have a small diameter (reflecting the small world property of real world social networks)~\cite{kleinberg00small,nmw00}.
But much less is known about the network properties that enable fast spreading of complex contagions. Complex contagions do not adhere to sub-modularity and sub-additivity upon which many analyses depend. Also, the super-additive character of complex contagions means that they are integrally related to community structure, as complex contagions intuitively spread better in dense regions of a network-- an observation concurred by real world experiments~\cite{centola2010spread}.

There have been only a few results on formal analysis of the spreading characteristics of complex contagions. All of them use the model of a $k$-complex contagion, in which time is divided into rounds and a node becomes infected (e.g., adopting the new behavior) in the next round if at least $k$ of its neighbors are infected in the current round.
Immediate questions to answer include whether a complex contagion spreads to the entire graph, and if so, how many rounds it will take. Despite the simplicity of this model, it sufficiently captures the qualitative difference of single versus cumulative exposure in social influences and already embraces a fair amount of technical challenges. This is also the model we adopt in this paper.

\paragraph{Prior Work}

Centola and Macy~\cite{G08} studied complex contagions in the Watts-Strogatz model.
The Watts-Strogatz model has nodes on a ring where nodes nearby on the ring are connected by edges and a small fraction of the edges are uniformly randomly `re-wired'. The network diameter before random rewiring is large (linear in the number of nodes) but with even a small number of randomly rewired edges the diameter quickly drops down. The strong community structure helps a complex contagion to spread but unfortunately the spreading is slow and cannot exploit the random edges that help to spread simple contagions. On the contrary, the random rewiring starts to erode the capability to support complex contagions as the community structure starts to break apart.

Ghasemiesfeh et.\ at~\cite{Ghasemiesfeh:2013:CCW} made this observation on the importance of the distribution of these random edges more rigorous, under the more general small world model proposed by Kleinberg~\cite{kleinberg00small}. In the 2D version of this model, nodes stay on an underlying  2D grid. Nodes that are within a constant Manhattan distance of each other are connected, these edges are denoted as
\emph{strong ties}, which model community structures.
In addition, each node on the grid issues a constant number of random edges following the distribution that $p$ chooses to connect to $q$ with probability proportional to $1/|pq|^{\gamma}$, where $|pq|$ is the Manhattan distance between $p, q$, and $\gamma$ is a non-negative parameter. These randomly chosen edges are labeled as \emph{weak ties} and help create the small world property and support fast spreading of simple contagions. When the initial seeds are chosen to be a cluster of nearby nodes in the underlying metric, it is proved in~\cite{Ghasemiesfeh:2013:CCW}, that for $\gamma = 2$, a $2$-complex contagion spreads in at most $O(\log^{3.5}n)$  rounds in a network of $n$ nodes, but for $\gamma = 0$ the contagion needs $\Omega(\poly(n))$ time to cover the entire network (this setting corresponds to a $2$-dimensional  Newman-Watts model~\cite{nmw00}).

\subsection{Our Results}
We substantially improve upon the prior work~\cite{Ghasemiesfeh:2013:CCW}  by characterizing $k$-complex contagions in the Kleinberg's small world model for \emph{all} values of $k$ and an almost complete range of $\gamma$, the exponent of a spatial distribution upon which random edges are created .

In addition, we also show that a subtle difference in the multiplicity of edges in Kleinberg's small world model implies large differences in what parameter regime quickly spreads $k$-complex contagions. In the model analyzed by~\cite{Ghasemiesfeh:2013:CCW}, it was implicitly assumed that the $k$ random edges placed by the same node are sampled \emph{without replacement}, thus disallowing multi-edges. In this case, to spread the contagion, we need at least $k$ initially infected nodes. In a slight variation, if the $k$ edges are chosen independently of each other, i.e., \emph{with replacement}, the generated graph may have multi-edges. Thus, a single node by itself may start a complex contagion. Both variations have real world interpretations. In the former variation, we need to have different infected neighbors to generate enough influence. In the latter, we count the number of repeated exposures to the new idea/belief, even if the exposure is from the same friend/entity. Analytically, however, this minor difference generates different behaviors. We show below that the parameter range for $\gamma$ to allow fast spreading of a complex contagion in each of these variations differs.

Let $\alpha_k =  \frac{2(k^2 + k + 2)}{k(k + 1)}$ and $\beta_k = \frac{2(k+1)}{k}$.  We show that $k$-complex contagions in the Kleinberg's small world model without multi-edges spread in $O(\polylog(n))$ rounds if $\gamma \in [2, \alpha_k)$, and in $\Omega(\poly(n))$ rounds otherwise (except for $\gamma=\alpha_k$ for which we do not know). We refer to polylogarithmic and polynomial speeds as \emph{fast}  and \emph{slow} respectively. For $k$-complex contagions in the model with allowing multi-edges, the fast spreading parameter range for $\gamma$ changes to $[2, \beta_k)$ instead, outside of which the contagion spreads slowly, again, except for $\gamma=\beta_k$  for which we do not know. This is summarized in Figure~\ref{fig:plot}.  We note that the results for $\gamma = 2$ and $\gamma = 0$ were already known by previous work~\cite{Ghasemiesfeh:2013:CCW}.

\begin{figure}[htbp]
\begin{center}
\includegraphics[scale=0.2]{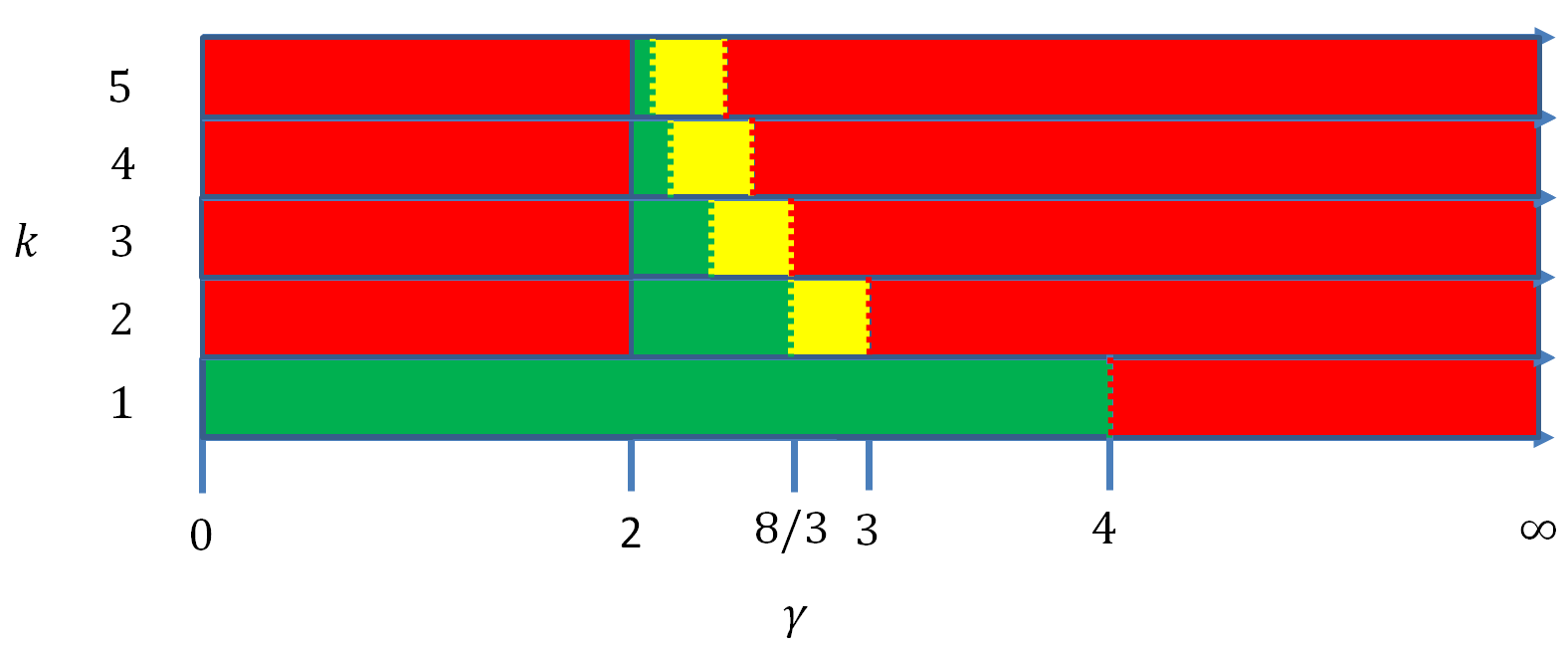}
\caption{Speed for a $k$-complex contagion for Kleinberg's small world model with parameter $\gamma$.  Green indicates polylogarithmic spreading in both models; red indicates that the contagions require polynomial number of rounds to spread in both models; and yellow indicates polylogarithmic spreading in the model with multi-edges and polynomial spreading in the model w/o multi-edges.} \label{fig:plot}
\end{center}
\end{figure}

We note that as $\gamma$ approaches $2$ from the right, it successfully spreads $k$-complex contagions for a greater range of $k$, however, if $\gamma > 2$, this range is always bounded.  Additionally, Figure~\ref{fig:plot} illustrates the results of $k = 1$ (the diameter), which were proven in a series of work~\cite{kleinberg00small,Mar04,Mar05}.  Our results naturally reprove these results for $\gamma > 2$, however, we note that in contrast to the simple contagions, complex contagions do not spread fast for any parameter in the range $[0, 2)$.

We give a conceptual view of our results next. Note that complex contagions require a \emph{wide bridge} to jump from the infected nodes to the uninfected ones. Such wide bridges are already supplied by the strong ties, allowing the complex contagion to spread along the underlying metric of the model. But such spreading is slow. To boost the speed, the weak ties must also form wide bridges. The ability to form such wide bridges is significantly affected by the changes in the parameter $\gamma$. In particular, when $\gamma$ is in the sweet spot interval, the weak ties form wide bridges in a variety of different length scales. This triggers recursive cascading and allows the contagion to grow in size by an exponential rate and thus, the speed of complex contagions becomes fast (polylogarithmic).
This intuition bears some similarity to the analyses of Kleinberg's myopic routing algorithm~\cite{kleinberg00small} and the diameter of Kleinberg's model~\cite{Mar04,Mar05}. In those analyses, if $\gamma$ is in a range where random edges are prevalent in all distance ranges, myopic routing is fast and the diameter is small (both polylogarithmic).

Our lower bound proofs, on the other hand, require a careful dissection of the spreading pattern of complex contagions in the graph. When $\gamma$ is smaller than the sweet spot interval, the weak ties do not form wide bridges (Theorem~\ref{thm:small-gamma}). And when $\gamma$ is larger than the sweet spot interval, the wide bridges formed by the weak ties are too short, and unable to reach far away regions of the network (Theorems~\ref{thm:lowerbound-too-short} and~\ref{thm:lowerbound-cloud-gen}). In both cases, it takes a polynomial number of rounds for the contagion to spread to the entire graph.

The results in this paper substantially enrich previous observations on how community structures are relevant to the spreading of complex contagions~\cite{centola2010spread}. While many dense community structures could spread complex contagions, certain edge distributions could make the spreading exponentially faster than others. The analysis in this paper shows the delicacy and richness of the complex contagion phenomenon.

\paragraph{Organization of This Paper}
In Section~\ref{sec:model}, we present the models and definitions.
Section~\ref{sec:upperbound} discusses the polylogarithmic upper bounds.
Section~\ref{sec:lowerbound-short} shows polynomial lower bounds. 
Finally, in Section~\ref{sec:related}, we outline some of the related works to our paper.


\section{Preliminaries} \label{sec:model}
We formally define a $k$-complex contagion process in a graph. We assume $k$ is a small constant.

\begin{definition}
A \defn{$k$-complex contagion} ${\rm CC}(G, k, \mathcal{I})$ is a contagion that initially infects vertices of $\mathcal{I}$ and spreads over graph $G$. The contagion proceeds in rounds. At each round, each vertex with at least $k$ infected neighbors becomes infected. The vertices of $\mathcal{I}$ are called the initial seeds.
\end{definition}

In the Kleinberg's small world model~\cite{kleinberg00small}, $n$ nodes are defined on a $\sqrt{n} \times \sqrt{n}$ grid\footnote{In order to eliminate the boundary effect, we wrap up the grid into a torus -- i.e., the top boundary is identified with the bottom boundary and the left boundary is identified with the right boundary.}. We connect each node $u$ to nodes within grid Manhattan distance $\lceil \sqrt{m} \rceil$, where $m$ is a constant. These edges are referred to as \defn{strong ties}. In addition, each node generates $m$ random outgoing edges, termed \defn{weak ties}. The probability that node $p$ connects to node $q$ via a random edge is equal to $\lambda /|pq|^\gamma$, in which $|pq|$ is the Manhattan distance of $p, q$ and $\lambda$ is a normalization factor. 
\begin{itemize} \denselist
\item When $\gamma>2$, $\lambda=\Theta(\gamma-2)$. That is, the probability that node $p$ chooses node $q$ as a neighbor through a weak tie is $\lambda /|pq|^\gamma=\Theta(1 /|pq|^\gamma)$.
\item When $0\leq \gamma<2$, $\lambda=\Theta(n^{\gamma/2-1})$. The probability that node $p$ chooses node $q$ as a neighbor through a weak tie is $\lambda /|pq|^\gamma=\Theta(n^{\gamma/2-1} /|pq|^\gamma)$.
\end{itemize}

We consider the directed graph model where the weak ties issued by a node $u$ has $u$ as the tail and strong ties are bidirectional. A $k$-complex contagion spreads in the inverse direction of an edge: a node becomes infected if it follows at least $k$ infected neighbors.

We consider two variations of the Kleinberg's small world model: first, \defn{$K^W_{m,\gamma}(n)$}, where the weak ties are chosen without replacement, that is, we do not permit multiple edges; and secondly,  \defn{$K^I_{m,\gamma}(n)$}, where the weak ties are distributed \emph{independently} (with replacements).  We will see that these two models behave very differently with respect to complex contagions.  We will use the notation \defn{$K^*_{m,\gamma}(n)$} to refer to both  $K^W_{m,\gamma}(n)$ and  $K^I_{m,\gamma}(n)$ simultaneously.

We say that two nodes $u, v$ are a \defn{seed pair} if they are adjacent on the grid structure of $K^*_{m,\gamma}(n)$.  We will say that $k$ nodes $u_1, \cdots, u_k$ are a \defn{$k$-seed cluster} if they form a connected subgraph via \emph{only} the grid structure.


\section{Upper Bounds} \label{sec:upperbound}
In this section, we prove two upper bound theorems. Let $\alpha_k =  \frac{2(k^2 + k + 2)}{k(k + 1)}$ and $\beta_k = \frac{2(k+1)}{k}$. We consider $k$-complex contagions on graphs $K^*_{m,\gamma}(n)$, where $2 \leq  k \leq m=O(1)$. We first prove that a $k$-complex contagion spreads in polylogarithmic number of rounds on $K^{W}_{m,\gamma}(n)$ when  $\gamma \in (2, \alpha_k)$. Using a similar technique, we prove a
polylogarithmic upper bound on the speed of $k$-complex contagions in  $K^{I}_{m,\gamma}(n)$ when  $\gamma \in (2, \beta_k)$. The case of $\gamma=2$ is handled for both of the models together (Theorem~\ref{thm:gamma-2}). We note that the subtle difference in how the networks are defined, makes a rather large difference in how complex contagions spread.

Both of the upper bound proofs are obtained by establishing the following recurrence relation: 
$$T(n) = k + 2T(n^{1-\delta}),$$ 
where $T(n)$ is the time it takes to infect a square of $n$ nodes in $K^{*}_{m,\gamma}(n)$.
 The idea behind the recursion is to start with a square, $S$, of $n$ nodes. Divide $S$ into $n^{\delta}$ many smaller squares for $0<\delta<1$. Fix a particular small square $A$ where the initial seeds locate. We then show that when $A$ is fully infected, there will be a new seed in each of the other $n^{\delta}-1$ small squares and thus they will be infected in parallel. The main difference for allowing versus disallowing multi-edges is that when multi-edges are allowed it is often enough to infect a single node $u$ to obtain a new \emph{seed}. This means that new seeds appear more often compared to the case where multi-edges are not allowed. This is because, with constant probability, $k-1$ of $u$'s local neighbors will have $k$ redundant edges to this vertex, thus allowing the $k$-complex contagion to spread thereafter.

\begin{theorem} \label{thm:upperbound-W-gen-k}
Let $ 2 \leq  k \leq m=O(1)$. Consider a Kleinberg's small world model $K^{W}_{m,\gamma}(n)$ of $n$ vertices where $2 < \gamma< \alpha_k$.
 Let 
 $ 0 < \delta < 1- \frac{\gamma}{\alpha_k}$,  $\frac{1-\delta}{\left ({k + 1 \choose 2} + 1 \right )(1 - \delta) + {k + 1 \choose 2}\gamma/2} \leq c $,  
 and $0 < d$, and $\alpha = {\frac{c}{2}+ \log_{\frac{1}{1-\delta}} 2}$ be constants.
 If we start a $k$-complex contagion from a $k$-seed cluster $\mathcal{I}$, it takes at most $O \left ( \log ^{\alpha} n  \right )$ rounds for the contagion to spread to the whole network with probability at least $1-n^{-d}$.
\end{theorem}

We first provide a definition:

\begin{definition} \label{def:rec-spreading} Fix constants $\delta, c, d, k$, let $\lambda$ be the normalization factor of $K^{*}_{m,\gamma}(n)$ and let $r = (\frac{6d}{\lambda})^c$.  We say that a $K^{*}_{m,\gamma}(n)$ model is \defn{$(\delta, c, d, k)$-recursively spreading} if whenever 
\begin{enumerate} \denselist
\item $S$ is an $\ell$-sized square ($\sqrt{\ell} \times \sqrt{\ell}$) of vertices  in  $K^{*}_{m,\gamma}(n)$ where $\ell > (r \log^c(n))^{(\frac{1}{1-\delta})}$; 
\item $A$ and $B$ are any two disjoint $\ell^{1-\delta}$-sized subsquares of $S$;  and
\item  $A$ is fully infected, then with probability at least $1-\ell^{2(1-\delta)}/n^d$, there is a new $k$-seed cluster in $B$ that is infected in $\leq k$ rounds. 
\end{enumerate} 
The probability is over the coin flips of the $K^{*}_{m,\gamma}(n)$ model.\end{definition}

We next state two lemmas.  Theorem~\ref{thm:upperbound-W-gen-k} follows directly from these lemmas.

\begin{lemma} \label{lem:upper-bound} Fix constants $\delta, c, d, k$.  If a $K^{*}_{m,\gamma}(n)$ model is $(\delta, c, d, k)$-recursively spreading , then if we start a $k$-complex contagion from a $k$-seed cluster, it takes at most $O \left ( \log ^{\alpha} n  \right )$ rounds for the contagion to spread to the whole network with probability at least $1-n^{-d}$, where $\alpha = {\frac{c}{2}+ \log_{\frac{1}{1-\delta}} 2}$, and the probability is over the coin flips of the $K^{*}_{m,\gamma}(n)$ model.
\end{lemma}

\begin{lemma} \label{lem:recurrence-W-gen-k}  Fix constants $2 < \gamma< \alpha_k$, $0 < \delta < 1- \alpha_k \gamma$,  $\frac{1-\delta}{\left({k + 1 \choose 2} + 1 \right)(1 - \delta) + {k + 1 \choose 2}\gamma/2} \leq c$, and $0 < d$, then
 $K^{W}_{m,\gamma}(n)$  is  $(\delta, c, d, k)$-recursively spreading.  \end{lemma}

\begin{proof}[Proof of Lemma~\ref{lem:upper-bound}]
Let $T(\ell)$ be the time it takes to infect all vertices on an $\sqrt{\ell} \times \sqrt{\ell}$ square in a $K^{*}_{m,\gamma}(n)$ graph starting from  a $k$-seed cluster  with probability $1- \ell^2/n^d$.  
Let $r = \left(\frac{6 d }{\lambda^3}\right)^{c}$ where $\lambda$ is the normalization factor (a constant) of the $K^{*}_{m,\gamma}(n)$ model.
We establish
the following recurrence relation for $T(\ell)$:
\begin{align}\label{eq:recursion}
T(\ell) \leq
\begin{cases}
k + 2T(\ell^{1-\delta}) & \text{if} \quad \ell >  (r \log^{c} n)^{(\frac{1}{1-\delta})}   \\
\sqrt{\ell}  & \text{o.w.}
\end{cases}
\end{align}

The second case of the recurrence follows immediately because a complex contagion using only strong ties can cover a cell of size $\sqrt{\ell}\times \sqrt{\ell}$ in $O(\sqrt{\ell})$ rounds. Next, we explain the first case in the recurrence.

Start with a square $S$ of $\ell > (r \log ^c n)^{(\frac{1}{1-\delta})}$ nodes.  Divide $S$ into $\ell^{\delta}$ many smaller squares, each with size $\ell^{1-\delta}$.
Fix a particular small square $A$ where the initial seeds locate.
By the definition of $T(\ell^{1-\delta})$, with probability at least $1- \ell^{2(1-\delta)}/n^d$, subsquare $A$ will be infected in time $T(\ell^{1-\delta})$.   By assumption of $(\delta, c, d, k)$-recursively spreading, after $k$ additional time steps, with probability at least $1- (\ell^\delta - 1) \ell^{2(1-\delta)}/n^d$ each of the other $(\ell^\delta - 1)$  sub-squares will have a $k$-seed cluster (applying the union bound).   And so with probability at least $1- (\ell^\delta - 1)\ell^{2(1-\delta)}/n^d$ they will all be infected in time $T(\ell^{1-\delta})$.  At this point every node of $S$ is infected.  By a union bound, the probability of failure is at most
$$\ell^\delta \left(2\frac{ \ell^{2(1-\delta)}}{n^d}\right) = 2\frac{\ell^{2-\delta}}{n^d} <\frac{\ell^2}{n^d} $$
The last inequality follows because $2 \leq \ell^{\delta}$.

We use the recurrence to upper bound $T(\ell)$. The height $h$ of the recurrence tree is such that
$$\left(r \log^c(n) \right)^{(\frac{1}{1-\delta})^h} > n.$$  Rearranging, we can see that $h \geq \frac{\log\log n - \log c \log\log(r \log n)}{-log(1-\delta)}$ and taking $h = \frac{\log\log n}{-\log(1-\delta)}  =   \log_{1/1-\delta}\log n \approx (1/\delta)\loglog n $ suffices.
The cost of the recurrence is dominated by the sum of the running times of the leaf cases. The branching factor of the tree is 2 and thus the number of leaves is at most  $2^{h}=2^{\log_{1/1-\delta}\log n}=(\log n)^{\log_{1/1-\delta} 2}$.
Therefore,
$$
T(n)=O \left( \sqrt{\log^{c} n} \right ) \cdot O \left (\left (\log n \right )^{ \log_{\frac{1}{1-\delta}} 2} \right )
=O \left (\left (\log n \right )^{\frac{c}{2}+ \log_{\frac{1}{1-\delta}} 2} \right ).
$$
\end{proof}
We present the proof of Lemma~\ref{lem:recurrence-W-gen-k} for $k=2$ here, because the presentation is easier and gives more insight into the nature of the lemma. The essence of the argument for $k \geq 3$ is the same and its proof is presented in Appendix~\ref{app:upperbound-lemmas}.
\begin{proof}[Proof of Lemma~\ref{lem:recurrence-W-gen-k} when $k=2$]
 When $k=2$, we take $2 < \gamma<8/3$, $0 < \delta < \frac{8- 3 \gamma}{8}$,  $\frac{2-2 \delta}{8-3\gamma - 8 \delta} \leq c$, and $0 < d$. Then we argue that $K^{W}_{2,\gamma}(n)$  is  $(\delta, c, d, 2)$-recursively spreading.

Fix  $\ell > (r \log ^c n)^{(\frac{1}{1-\delta})}$, and $\ell$-sized square $S$ and two disjoint $\ell^{(1-\delta)}$-sized subsquares $A$ and $B$, and partition $B$ into $\ell^{(1-\delta)}$/2 disjoint pairs of (grid) neighbor nodes  $u, v$. Assume that $a \in A$. Let $d(u,a)$ be the Manhattan distance of $u$ and $a$.
 \begin{enumerate}[i)] \denselist
 \item $P_{1}(\ell)=\Prob{\text{$u$ has a weak tie to $a \in A$, via a particular edge}}$
 \item $P_{2}(\ell)=\Prob{\text{$u$ has a weak tie to $A$, via a particular edge}} \geq |A| \times P_{1}(\ell)$
 \item $P_{3}(\ell)=\Prob{\text{$u$ is connected to $A$, via $2$ distinct weak ties}} $
 \item $P_{4}(\ell)=\Prob{\text{$u,v$ form a \emph{new seed} in $B$}} $
  \item$P_{5}(\ell)=\Prob{\text{a new seed forms in $B$}} $
 \end{enumerate}
 \begin{align*}
P_{1}(\ell) & \geq  \frac{\lambda}{(d(u,a))^{\gamma}}  =   \frac{\lambda}{ \sqrt{\ell^{\gamma}} } ; \\
  P_{2}(\ell) &\geq |A| \times P_{1}(\ell)  = \lambda \ell^{1-\delta-\gamma/2} ;\\
 P_3(\ell) & \geq |\{ (a,b) | a, b \in A \}| \times \Prob{\text{$u$ has a weak tie to $a$ and a weak tie to $b$}}  \\
 &\geq  {|A| \choose 2} P_{1}(\ell)^{2}  \geq \frac{ \lambda^2 \ell^{2-2\delta-\gamma}}{3};\\
 P_{4}(\ell) & \geq P_{2}(\ell)\times P_{3}(\ell)  = \frac{ \lambda^3  \ell^{3-3\delta-3\gamma/2} }{3} ;\\
P_{5}(\ell)  &\geq 1 - (1 - P_{4}(\ell))^{|B|/2} \geq  1 - \exp\{P_{4}(\ell) |B|/2\} = 1 - \exp\{-\frac{\lambda^3}{6}\ell^{4-4\delta -3\gamma/2}\}
\end{align*}
We want $P_5(\ell)$ to be very close to $1$ and that means that we want the power of $\ell >0$, which
gives us:
$$
4-4\delta -3\gamma/2 >0  \ \Rightarrow \ \gamma<8/3,  \ \delta< \frac{8-3\gamma}{8}.
$$
Also, note that $P_5$ is increasing in $\ell$. Therefore, the smallest probability happens when
then size of $\ell$ is the smallest; that is $\ell^{1-\delta} =r \log^c n$ and $\ell = (r \log^c n)^{1/(1-\delta)}$.
We want $1-P_{5}\left((r \log^c n)^{1/(1-\delta)} \right)$ to be polynomially small even when $\ell$
takes on its smallest value.
This requires that the power of $\log n$ be $ \geq 1$:
$$
4c -\frac{3c\gamma}{2-2\delta} \geq 1  \  \Rightarrow \ c \geq  \frac{2-2 \delta}{8-3\gamma - 8 \delta};
$$
and by replacing $r=(6d/\lambda^3)^c$, we get that
$$
P_{5}(\ell)  \geq 1 - \exp\{-\frac{\lambda^3}{6} \left(  (r \log^c n)^{1/(1-\delta)} \right)^{4-4\delta -3\gamma/2}\} \geq 1-\exp\{-d \log n\} \geq 1-n^{-d}.
$$
Finally, notice that this probability is much stronger than the statement of Lemma.
\end{proof}

Now, we prove a polylogarithmic upper bound for the speed of a $k$-complex contagion on $K^I_{k,\gamma}(n)$ when $2 < \gamma<\beta_k$ where $\beta_k = \frac{2(k+1)}{k}$. We note that the subtle difference in how the networks are defined, makes a rather large difference in how complex-cascades spread.

Consider the case of $k=2$. 
The main difference with Theorem~\ref{thm:upperbound-W-gen-k} is that new configurations of nodes/edges
can act as \emph{new seeds} in $K^I_{2,\gamma}(n)$ compared to $K^W_{2,\gamma}(n)$. In $K^W_{2,\gamma}(n)$,
a new seed appears in the graph if two grid neighbor nodes $u,v$ have a total of $3$ randomly created edges (weak ties) to the infected parts of the network. One of them, say $u$
becomes infected first using its two weak ties and $v$ becomes infected using one weak tie and one grid edge.
However, in $K^I_{2,\gamma}(n)$, if a node $u'$ has by itself $2$ randomly created edges (weak ties) to the infected part, and a grid neighbor of $u'$, called $v'$, has one randomly created edge to $u'$ besides the grid edge to $u'$, then $u',v'$ can constitute a new seed together.

\begin{theorem} \label{thm:upperbound-I-gen-k}
Let $ k \leq m=O(1)$. Consider a Kleinberg's small world model $K^{I}_{m,\gamma}(n)$ of $n$ vertices where $2 < \gamma< \beta_k$.
Let $0 < \delta < 1- \frac{\gamma}{\beta_k}$,  $\frac{1- \delta}{(k + 1)(1 -\delta) - k \gamma/2} \leq c$, and $0 < d$, and $\alpha = {\frac{c}{2}+ \log_{\frac{1}{1-\delta}} 2}$ be constants.
If we start a $k$-complex contagion from a $k$-seed cluster $\mathcal{I}$, it takes at most $O \left ( \log ^{\alpha} n  \right )$ rounds for the contagion to spread to the whole network with probability at least $1-n^{-d}$.
\end{theorem}

The proof of Theorem~\ref{thm:upperbound-I-gen-k} follows immediately from Lemma~\ref{lem:upper-bound} and  Lemma~\ref{lem:recurrence-I-gen-k} (proved in Appendix~\ref{app:upperbound-lemmas}).  The difference is that the bounds on values of $\gamma, \delta, c, r$ in Lemma~\ref{lem:recurrence-I-gen-k} vary from those in Lemma~\ref{lem:recurrence-W-gen-k}. 

\begin{lemma} \label{lem:recurrence-I-gen-k}  Fix constants $2 < \gamma< \beta_k$, $0 <  \delta < 1- \frac{\gamma}{\beta_k}$,  $\frac{1- \delta}{(k + 1)(1 -\delta) - k \gamma/2} \leq c$, and $0 < d$, then
 $K^{I}_{m,\gamma}(n)$  is  $(\delta, c, d, k)$-recursively spreading.  \end{lemma}

\subsection{When $\gamma = 2$}

When $\gamma = 2$, an upper bound of $O(\log^{3.5} n)$ was proven in~\cite{Ghasemiesfeh:2013:CCW} for a $2$-complex contagion on the Kleinberg's small world model $K^W_{m, 2}(n)$. This bound can be slightly improved since we observe that for a new seed pair to be generated, we only need a total of three edges to the infected nodes, rather than four edges assumed in~\cite{Ghasemiesfeh:2013:CCW}. Thus, the upper bound can be improved to be $O(\log^{3} n)$.

In~\cite{Ghasemiesfeh:2013:CCW} the result has also been extended to $k$-complex contagions ($k >2 $) on $K^W_{m, 2}(n)$ and an upper bound of $O(\log^{k^2/2+1.5})$ can be obtained (there was a small typo in~\cite{Ghasemiesfeh:2013:CCW}).
Again, we can slightly improve this upper bound: in order to get a new $k$-seed cluster $(u_{1}, ..., u_{k})$, we do not need all $u_{i}$'s, $1\leq i\leq k$, to have $k$ edges to the infected nodes. Since $(u_{1}, ..., u_{k})$ is a clique, it is sufficient for $u_{1}$ to have $k$ random edges to the infected nodes, $u_{2}$ to have $k-1$ random edges to the infected nodes, etc. Thus, the total number of  random edges needed from $(u_{1}, ..., u_{k})$ to the infected nodes would be equal to $k+(k-1)+...+1=k(k+1)/2$. The idea of the proof works for obtaining an upper bound on the speed of a $k$-complex contagion on the other Kleinberg's small world model variation $K^I_{m, 2}(n)$ as well. Thus, we get the following theorem:

\begin{theorem} \label{thm:gamma-2}
Consider a Kleinberg's small world model of $n$ vertices where $ \gamma = 2$, $K^*_{m, 2}(n)$,  (e.g.\ $K^I_{m, 2}(n)$ or $K^W_{m, 2}(n)$).
If we start a $k$-complex contagion from a $k$-seed cluster, it takes at most $O \left ( \log^{k(k+1)/4+1.5} n  \right )$ rounds for the contagion to spread to the whole network with probability $1-O(1/n)$.
\end{theorem}

\subsection{When $\gamma = \alpha_k, \beta_k$}  The only points of parameter $\gamma$ that we are not able to classify are $\alpha_k$ in the $K^W_{m, \gamma}(n)$  model and $\beta_k$ in the $K^I_{m, \gamma}(n)$ model. We note that in both Theorems~\ref{thm:upperbound-W-gen-k} and~\ref{thm:upperbound-I-gen-k}, the exponent of the logarithm in the time to infection, grows toward infinity.  We conjecture that this is necessary and that the cascade at the critical points of $\alpha_k$ and $\beta_k$ is not polylogarithmic, though we cannot prove this.  This is in contrast to $\gamma = 2$ where a transition also occurs.


\section{Lower Bounds}\label{sec:lowerbound-short}
In this section, we first describe our three lower bound theorems and the general idea behind each of them. 
Then, we present the detailed proofs.

What enabled us to prove the polylogarithmic upper bounds
in Section~\ref{sec:upperbound}, was the abundance of weak ties in all the distance ranges: from short weak ties of constant length to long weak ties of length $\Omega(\sqrt{n})$.
Such an abundance ceases to exist when $\gamma$ is outside of $[2,\alpha_k)$ for $K^{W}_{m,\gamma}(n)$ and outside of $[2, \beta_k)$ for $K^{I}_{m,\gamma}(n)$. When $\gamma \in [0,2)$ weak ties become too random. In Subsection~\ref{sec:lowerbound-random}, we show that too much randomness causes the complex contagion to be considerably slower. The intuition is that there is lack of coherence to enable complex contagion to generate \emph{new seeds} until the contagion has grown to a polynomially large portion of the graph. The following theorem works for both graph models and for all possible values of $k$.

\begin{theorem}\label{thm:small-gamma}
Consider the  Kleinberg's small world model $K^*_{m,\gamma}(n)$ of $n$ vertices where $0\leq\gamma<2$. If we set $\mathcal{I}$ to be a $k$-seed cluster, then the speed of a $k$-complex contagion ${\rm CC}(K^*_{m,\gamma}(n), k, \mathcal{I})$ is at least $n^{\delta}$, $\delta <\min \left( \frac{k-2+\eps/2}{2k},\frac{k-\gamma/2-1+\eps}{2+2k-k\gamma} \right) $, with probability
at least $1-{\rm O} (n^{1-\eps})$, where $m\geq k$ is a constant and $0<\eps<1$.
\end{theorem}

When $\gamma >2$, there is a more subtle situation. As $\gamma$ becomes big, weak ties become shorter. When the weak ties become too \emph{short}, one might suspect that the complex contagion might take a long time to travel \emph{long} distances in the graph. Our first attempt in formalizing this intuition appears in Theorem~\ref{thm:lowerbound-too-short} below.  The theorem applies to $k$-complex contagions on both variations of the Kleinberg's model (with and without multi-edges) for $\gamma \in (\beta_k, \infty)$, where $\beta_k = \frac{2(k+1)}{k}$.  The proof divides the graph into a polynomial number of blocks (of polynomial size) and shows that the $k$-complex contagion will travel only through adjacent blocks.  We show that if adjacent blocks of a block, like $B$, are not already infected, with high probability no node inside that block, block $B$, gets infected.

\begin{theorem} \label{thm:lowerbound-too-short}
Consider a $k$-complex contagion ${\rm CC}(K^*_{m,\gamma}(n), k, \mathcal{I})$ on Kleinberg's small world network $K^*_{m,\gamma}(n)$, and let $\mathcal{I}$ be a $k$-seed cluster.  For every $k$ and for every $\gamma > \beta_k = \frac{2(k+1)}{k}$, there exists $\delta, \beta > 0$ such that for any constant $m > k$, with probability $1 - 1/n^{\beta}$, it takes $\Omega(n^ \delta)$ rounds for ${\rm CC}(K^*_{m,\gamma}(n), k, \mathcal{I})$ to spread to all the network.  \end{theorem}

We have proved an upper bound in $K^{W}_{m,\gamma}(n)$, the model which does not allow multi-edges, for $\gamma\in [2, \alpha_k)$ and an upper bound in $K^{I}_{m,\gamma}(n)$, the model which allows multi-edges, for $\gamma \in [2, \beta_k)$. Theorem~\ref{thm:lowerbound-too-short} gives a lower bound for the range of $\gamma\in(\beta_k ,\infty)$ for both $K^{W}_{m,\gamma}(n)$ and $K^{I}_{m,\gamma}(n)$ models.

We are left with one case: What is the speed of $k$-complex contagions in $K^{W}_{m,\gamma}(n)$ when $\gamma \in (\alpha_k, \beta_k]$? This proves to be the trickiest case of them all, because it turns out that when $\gamma \in (\alpha_k, \beta_k]$, long ties are still abundant in the $K^{W}_{m,\gamma}(n)$ graph. This means that infections can occur between parts of the graph that are far from each other. And that is why the argument of Theorem~\ref{thm:lowerbound-too-short} can not be extended to solve the case of $\gamma \in (\alpha_k, \beta_k]$ for
$K^{W}_{m,\gamma}(n)$.

A better dissection of the nature of a complex contagion is needed to solve this case.
We prove a polynomial lower bound for the $K^{W}_{m,\gamma}(n)$ model when  $\gamma \in (\alpha_k, \infty) = (\alpha_k, \beta_k] \cup (\beta_k, \infty)$.
We achieve this by carefully analyzing the possible scenarios that could lead to generation of \emph{new seed clusters} in far-away distances. We show that despite the existence of infections occurring over large distances, new seed clusters
do not appear in these distances with high probability.

\begin{theorem}\label{thm:lowerbound-cloud-gen}
Consider the Kleinberg's small world model $K^{W}_{m,\gamma}(n)$ (without multi-edges) of $n$ vertices where $\frac{2(k^2 + k + 2)}{k(k + 1)}=\alpha_k<\gamma $. Let $0 < \epsilon < \frac{1 - \alpha_k /  \gamma }{2}$,
and $1 + {k+1 \choose 2}(1 - \gamma(1/2 -  \epsilon)) < \zeta < 0$.
Start a $k$-complex contagion from a $k$-seed cluster  $\mathcal{I}$.  Then, with probability at least $1- O(n^{\zeta})$, it takes the contagion at least $n^\epsilon/3$ rounds to spread to the whole network.
\end{theorem}


\subsection{Proof of Theorem~\ref{thm:small-gamma}}\label{sec:lowerbound-random}

Recall that when $0\leq \gamma<2$, the normalization factor $\lambda=\Theta(n^{\gamma/2-1})$. That is, the probability that node $p$ chooses node $q$ as a neighbor through a random edge is $\lambda /|pq|^\gamma=\Theta(n^{\gamma/2-1} /|pq|^\gamma)$.


Recall that in network $K^*_{m,\gamma}(n)$ each node has $m$ strong ties to the $m$ closest nodes in the grid, and $m$ weak ties randomly distributed to other nodes by the specified distribution. Assume that $s$ is a node in the graph. The initial seeds are a $k$-seed cluster which includes $s$ along with the closest $k$ nodes on the grid to $s$. Consider the set $D$ of nodes within Manhattan distance $n^{\delta}$ from $s$. 
We prove that the contagion inside $D$ would not utilize weak ties and can only spread along the strong ties until $D$ is completely covered. Since propagation along strong ties is local, it would require $\Omega(n^\delta)$ rounds to just cover $D$.

We say a node $u$ has a wide bridge to $D$, if  $u$ issues at least $k$ weak ties to nodes in $D$. We are going to show the following two statements:
\benum
\item The expected number of nodes in the annulus $A(2n^{\delta}, \sqrt{n})$, defined as the set of nodes whose Manhattan distance to $s$ is within the range $(2n^{\delta}, \sqrt{n}]$, with wide bridges to $D$, denoted as $Z_1$, is small;
\item The expected number of nodes in the disk $B(2n^{\delta})$, defined as the nodes whose Manhattan distance is at most $2n^{\delta}$ from $s$, that have wide bridges to $D$, denoted as $Z_2$, is small.
\eenum

First, we compute $Z_1$. Let $u$ be a node in the annulus $A(2n^{\delta}, \sqrt{n})$.
Then the distance of $u$ to any node $q$ inside $D$, $|uq|$, is bounded by constant factors of the distance from $u$ to the center $s$, $|us|$. That is, $c_1 |us| < |uq|< c_2 |us|$, with two constants $c_1 <c_2$. Let $y=|us|$ be the Manhattan distance of $u$ to $s$. One can easily verify that
\begin{align*}
&Q_1= \Prob{\text{Node $u \in A(2n^{\delta}, \sqrt{n})$ has 1 weak tie to $D$}}= O \left( \frac{\lambda }{y^\gamma } n^{2\delta}\right) \\
&Q_2= \Prob{\text{Node $u \in A(2n^{\delta}, \sqrt{n})$ has at least $k$ weak ties to $D$}}= O \left( \left( \frac{\lambda }{y^\gamma } n^{2\delta} \right)^k \right).
\end{align*}
We explain the second equation. 
Note that the outgoing edges of a node are independent of each other in $K^{I}_{m,\gamma}(n)$, but not
in $K^{W}_{m,\gamma}(n)$. While the out going edges of a node are not completely independent, they are close enough:  No matter what configuration the first (up to) $j-1$ ties have, the next tie has a $O(Q_1)$  probability of connecting to $D$.  Thus, the probability that any given vertex has $j$ long-range ties is bounded by $O((Q_1)^j)$.

Now, we compute $Z_1$ using linearity of expectation:
\begin{align*}
 Z_1 =\int_{2n^{\delta}}^{\sqrt{n}} O\left( \left( \frac{\lambda}{y^{\gamma}} \cdot n^{2\delta} \right )^k \right) \Theta( y) dy &=O \left( \lambda^k n^{2k\delta} \right) \int_{2n^{\delta}}^{\sqrt{n}} y^{1-k\gamma} dy \\
 & =
\begin{cases}
{\rm O} \left( 1/n^{k-1-2k\delta} \right), 0\leq \gamma<2/k\\
{\rm O} \left( \log n/n^{k-1-2k\delta} \right), \gamma=2/k\\
{\rm O} \left( 1/n^{k(\frac{1}{2}-\delta)(2-\gamma)-2\delta}\right), 2/k<\gamma<2
\end{cases}
 \end{align*}

Now, we bound $Z_2$. Consider a node $u$ in $B(2n^{\delta})$. Notice that $D$ is contained in the disk $D'$ of radius $3n^{\delta}$ from $u$. Thus, the probability that one weak tie issued by $u$ falls inside $D$ is no greater than the probability $P$ of this weak tie connecting to nodes in $D'$. The latter can be bounded from above:
$$P\leq \int_1^{3n^{\delta}} \frac{\lambda}{y^{\gamma}} ydy =\lambda (3n^{\delta})^{2-\gamma} =
O \left ( 1/n^{(\frac{1}{2}-\delta)(2-\gamma)} \right).$$
Hence, we get that the probability that node $u \in B(2n^{\delta})$ has $k$ weak tie to $D$ is at most $P^k$. We bound $Z_2$:
\begin{align*}
Z_2 ={\rm O} \left ( (2n^{\delta})^2 /n^{k(\frac{1}{2}-\delta)(2-\gamma)} \right ) ={\rm O} \left ( 1/n^{k(\frac{1}{2}-\delta)(2-\gamma)-2\delta} \right )
\end{align*}

We can verify that $Z_1=O(1/n^{1-\eps})$ and $Z_2=O(1/n^{1-\eps})$ for $0<\eps<1$ when
$$\delta <\min \left( \frac{k-2+\eps/2}{2k},\frac{k-\gamma/2-1+\eps}{2+2k-k\gamma} \right) .$$
Now, by Markov inequality, $Z_1 \geq 1$ and $Z_2 \geq 1$ only happen with small probability.
By union bound, we have that $Z_1=Z_2=0$ with probability at least
$1-{\rm O} (1/n^{1-\eps})$.

If there are no wide bridges to/from $D$, the contagion inside $D$ can only utilize
the strong ties and it will take it $\Omega(n^{\delta})$ rounds to cover $D$.
Hence, the contagion speed is $\Omega \left (n^{\delta}\right)$ with probability at
least $1-{\rm O} (1/n^{1-\eps})$.

\subsection{Proof of Theorem~\ref{thm:lowerbound-too-short}}
Fix $k$ and $\gamma > \beta_k$.  Define $\epsilon > 0$ so that $\gamma = \frac{2(k+\epsilon + 1)}{k}$.  Let $\delta = \frac{\epsilon}{4(1 + \epsilon)}$, $\beta = \epsilon/3$ and then fix a constant $m$.


The basic proof strategy is to divide the 2-dimensional $\sqrt{n} \times \sqrt{n}$ grid into $n^{\frac{1}{2} - \delta} \times n^{\frac{1}{2} - \delta}$ sized blocks. We call two blocks adjacent if they are next to each other, or are diagonally adjacent.  Any two vertices for two non-adjacent blocks are of distance at least $n^{\frac{1}{2} - \delta}$. We show that the probability that any vertex in a block has $k$ neighbors in non-adjacent blocks on the grid is at most $n^{-\beta}$. We then show that the $k$-complex contagion must spread via adjacent blocks. If this is the case, then it will take time $\Omega(n^{1/2 - \delta})$ to cross from the block containing the initial seeds to the block furthest away on the grid (wrapped up to a torus).

We will use the following two facts:
\benum
\item The probability that a weak tie issued by a node $u$ has Manhattan distance greater than $n^{\frac{1}{2} - \delta}$ is at most $P_1=\int_{n^{\frac{1}{2} - \delta}}^{\sqrt{n}} x^{-\gamma} x dx = O(1/n^{(\frac{1}{2} - \delta)(\gamma - 2)}).$
\item The probability that any vertex has $k$ neighbors of distance greater than $n^{\frac{1}{2} - \delta}$ is at most
$$O(1/n^{k(\frac{1}{2} - \delta)(\gamma - 2)-1}).$$

The outgoing edges of a node are independent of each other in $K^{I}_{m,\gamma}(n)$, but not
in $K^{W}_{m,\gamma}(n)$. While the out going edges of a node are not completely independent, they are close enough: No matter what configuration the first (up to) $j-1$ ties have, the next tie has a $O(P_1)$  probability of being longer than $n^{\frac{1}{2} - \delta}$.  Thus, the probability that any given vertex has $j$ long-range ties is bounded by $O((P_1)^j)$.
Consider a vertex $u$ which issues $m$ weak ties.  Each one has probability at most $O(1/n^{(\frac{1}{2} - \delta)(\gamma - 2)})$ of being longer than $n^{\frac{1}{2} - \delta}$.  Thus, the probability that $k$ weak ties are longer than $n^{\frac{1}{2} - \delta}$ is at most $$O \left ({m \choose k}(1/n^{(\frac{1}{2} - \delta)(\gamma - 2)})^k \right)= O \left (1/n^{(\frac{1}{2} - \delta) k(\gamma - 2)} \right).$$
Taking a union bound over $n$ vertices, the claim is proved.
\eenum

Now, by manipulations we see that $O(1/n^{(\frac{1}{2} - \delta)(\gamma - 2)-1}) \subseteq O(1/n^{\beta})$ because $\gamma - 2 = \frac{2 + 2 \epsilon}{k}$ and $\delta = \frac{\epsilon}{4(1 + \epsilon)}$ so we get $O(1/n^{(\frac{1}{2} - \frac{\epsilon}{4(1 + \epsilon)})(2 + \epsilon)-1})$ and after manipulations, this is $O(1/n^{\frac{\epsilon}{2}}) \subseteq O(1/n^{\beta})$.

If no vertex has $k$ weak ties of distance at least $n^{1/2-\delta}$, in one round the $k$-complex contagion can only spread to blocks that are adjacent to those blocks already containing infected vertices.  By the above claim, this is true with probability at least $1-n^{-\beta}$. In this case, the number of steps required is at least $\frac{\sqrt{n}}{(2k+1)n^{1/2-\delta}} > \frac{n^{\delta}}{3k}$.  Initially, at most $k$ rows of the grid are infected.  At each step, the infection can only move one row to the right or left of the infected regions.  After $r$ rounds there are at most $(2r + 1)k$ infected rows.  So, to infect $n^{\delta}$ rows, it takes $\Omega(n^{\delta}/(2k+1))$ rounds.
\qed

\subsection{Proof of Theorem~\ref{thm:lowerbound-cloud-gen}}
We say that a weak tie is \defn{long} if its length is at least  $n^{\frac{1}{2}-\epsilon}$ and it is \defn{short} otherwise. The definition of short ties also includes \emph{grid edges}. Let $G=K^{W}_{m,\gamma}(n)$ denote the graph.

\begin{definition}
We create a \defn{DAG, $D(G, A)$,} out of the route of infection on a network $G$ starting from $A$:  There is an \defn{edge from any node (not in $A$) to the $k$ nodes that \emph{first caused} it to be infected}. In the case that a node had more than $k$ neighbors at the time of infection, arbitrarily choose $k$.

For any subset of nodes $S$, define \defn{$\mathcal{A}(S)$ to be the nodes reachable from $S$ using \emph{only short ties} in the $D(G, A)$}. $\mathcal{A}(S)$ includes $S$ and the edges must be traversed in accord with their orientation.

Let \defn{$\mathcal{L}(S)$ be the number of long range ties connected to nodes in $\mathcal{A}(S)$}.  We overload notation and use $\mathcal{A}$ and $\mathcal{L}$ on singleton vertices.
\end{definition}

\begin{proposition} \label{claim:longest}
Let nodes $u$ and $v$ be connected by a (directed) path in $D(G, A)$; the time difference between when $u, v$ are infected is at least the length of the path.
\end{proposition}
\begin{proof}
Follows from the fact that any two nodes connected by an edge were infected $\geq 1$ round apart.
\end{proof}

\begin{lemma}\label{lem:eitheror}
Let a cascade start from $k$-seed cluster $A$ and let $B$ be another $k$-seed cluster:
\begin{itemize} \denselist
\item Either $A \cap \mathcal{A}(B) \neq \emptyset$;
\item or there exists a connected subset of nodes of size $k^2-k+1$ in $K^{W}_{m,\gamma}(n)$, that has at least ${k+1 \choose 2}$ \emph{long} ties connected to its vertices.
\end{itemize}
\end{lemma}

The following lemma is based on computation and its proof is presented at the end of this subsection.
 \begin{lemma}\label{lem:shortage-long-ties-on-paths-gen}
Let $\alpha_k<\gamma $,  $0 < \epsilon < \frac{1-\alpha_k/\gamma}{2}$, and $1 + {k+1 \choose 2}(1 - \gamma(\frac{1}{2} - \epsilon)) < \zeta < 0$.  Then with probability at least $1- O \left (  n^{\zeta} \right )$, no connected subset of nodes of size $k^2-k+1$ in $K^{W}_{m,\gamma}(n)$, has at least ${k+1 \choose 2}$ \emph{long} ties connected to its vertices.
\end{lemma}

We now use Proposition~\ref{claim:longest}, Lemma~\ref{lem:eitheror}, and Lemma~\ref{lem:shortage-long-ties-on-paths-gen} to prove Theorem~\ref{thm:lowerbound-cloud-gen}.

\begin{proof}[Proof of Theorem~\ref{thm:lowerbound-cloud-gen}]
Assume that the $k$-complex contagion starts from $k$-seed cluster $A$.  Let $B$ be a $k$-seed cluster that is distance $\sqrt{n}/3$ from $A$.  By Lemma~\ref{lem:shortage-long-ties-on-paths-gen}, we know that with probability $1- O \left ( n^{\zeta} \right )$, no connected subset of nodes of size $k^2-k+1$ in $G$ has at least ${k+1 \choose 2}$ \emph{long} ties connected to its vertices.  However, if this is the case, then by Lemma~\ref{lem:eitheror} we know that $A \cap \mathcal{A}(B) \neq \emptyset$.

If $A \cap \mathcal{A}(B) \neq \emptyset$, then there is a path of short edges from $A$ to $B$ over which the contagion is transmitted.  However, because $A$ and $B$ are distance $\sqrt{n}/3$ apart, such a path should travels distance $\sqrt{n}/3$ using edges that span at most distance $n^{\frac{1}{2} - \epsilon}$ and thus has at least $n^{\epsilon}/3$ edges.  $A$ is infected at time 0, so by Proposition~\ref{claim:longest} this implies that $B$ is not infected until round $n^{\epsilon}/3$.

Thus, with probability  $1- O \left ( n^{\zeta}\right )$, it takes  $ \geq n^{\epsilon}/3$ rounds to infect the entire graph.
\end{proof}
It remains to prove  Lemma~\ref{lem:eitheror} and the the following Proposition will help.

\begin{proposition} \label{lem:DAG-long-ties-count}
Fix  $s\leq k$ and let $S \subseteq \mathcal{A}(B)$ where $|\mathcal{A}(S)| \geq s$.  If $A \cap \mathcal{A}(B) = \emptyset$, then
$$\mathcal{L}(S) \geq \sum^{s-1}_{i =0 } (k-i) =  k+(k-1)+...+(k-s + 1).$$
\end{proposition}

\begin{proof}[Proof of Proposition~\ref{lem:DAG-long-ties-count}]

The proposition is true for $s$ = 0.  For the sake of induction, assume the proposition is true for when $s \leq  \ell < k$.
Let $|\mathcal{A}(S)| \geq \ell +1$.  Topologically order the vertices in $\mathcal{A}(S)$ and let $T$ be the set of the first $\ell$ vertices in this ordering and $v$ be the $\ell +1$st vertex in this ordering.

Then $|\mathcal{A}(T)| \geq \ell$ and so by the  inductive hypothesis, we have that $\mathcal{L}(T) \geq \sum^{\ell - 1}_{i =0 } (k-i) =  k+(k-1)+...+(k-\ell + 1)$.  However, $v$ has at most $\ell$ vertices before it in the ordering.  Thus, it necessarily has $k - \ell$ long-range ties because it required $k$ neighbors to be infected, but had at most $\ell$ short-range neighbors that contributed to this $k$.  So $ \mathcal{A}(T \cup \{v\}) = \mathcal{A}(T) \cup \{v\}  \subseteq  \mathcal{A}(S) $ has $k+(k-1)+...+(k-\ell + 1) + (k - \ell)$ long-range ties.  And this completes the inductive step.
\end{proof}

\begin{proof}[Proof of Lemma~\ref{lem:eitheror}]
Assume that  $A \cap \mathcal{A}(B) = \emptyset$.  We must find a connected subset of nodes of size $k^2-k+1$ in $G$, that has at least ${k+1 \choose 2}$ \emph{long} ties connected to its vertices. We consider two cases:
\begin{enumerate}
\item There exists a node $ v \in B$, such that $|\mathcal{A}(v)| \geq k$. \\
Sort $\mathcal{A}(v)$ in topological ordering
and consider the first node in the ordering, $v_0 \in \mathcal{A}(v)$, such that $|\mathcal{A}(v_0) |\geq k$.

We will show $\mathcal{A}(v_0)$ is a connected subset of nodes of size at most $k^2-k+1$ in $K^{W}_{m,\gamma}(n)$, that has at least ${k+1 \choose 2}$ \emph{long} ties connected to its vertices.
By $v_0$'s minimality, we have that $\forall u \in \mathcal{A}(v_0), \  |\mathcal{A}(u)| \leq k-1$.
Since the in-degree of $v_0$ is at most $k$, we have that
$$|\mathcal{A}(v_0)|  \leq  k (k-1) +1 = k^2 -k +1.$$
Also, because $|\mathcal{A}(v_0)| \geq k$, by Proposition~\ref{lem:DAG-long-ties-count}, we know that
$$\mathcal{L}(x_{0}) \geq k+(k-1)+\cdots + 1 +0 = {k+1 \choose 2}.$$

\item For all $v \in B$, $|\mathcal{A}(v)| \leq k-1$. \\
We will show $\mathcal{A}(B)$ is  a connected subset of nodes of size at most $k^2-k+1$ in $K^{W}_{m,\gamma}(n)$, that has at least ${k+1 \choose 2}$ \emph{long} ties connected to its vertices.
Now we have that $k \leq |\mathcal{A}(B)| \leq k(k-1)$.  The lower bound follows because $|B| = k$ and $B \subseteq \mathcal{A}(B)$.  The upper bound follows because $\mathcal{A}(B) = \bigcup_{v \in B} \mathcal{A}(v)$.
Because of the lower bound, by Proposition~\ref{lem:DAG-long-ties-count} , we know that
$$\mathcal{L}(B) \geq k+(k-1)+\cdots + 1 ={k+1 \choose 2} .$$
\end{enumerate}
Hence, in either cases, we find a connected subset of nodes of size at most $k^2-k+1$ in $G$, that has at least ${k+1 \choose 2}$ \emph{long} ties connected to its vertices.
\end{proof}

It only remains to prove Lemma~\ref{lem:shortage-long-ties-on-paths-gen}.



\begin{proof}[Proof of Lemma~\ref{lem:shortage-long-ties-on-paths-gen}]
Consider a specific set of $k^2-k+1$ connected nodes in $K^{W}_{m,\gamma}(n)$, $S=\{u_{1}, \dots, u_{k^2-k+1} \}$.
 We compute the probability that nodes of $S$  have collectively ${k+1 \choose 2}$ long range 
 ties to other nodes in $K^{W}_{m,\gamma}(n)$.
\begin{enumerate} [i)] \denselist
\item $P_{1}=\Prob{\text{Node $u_1$ has one long tie to any node in $K^{W}_{m,\gamma}(n)$}} $
\item $P_{2}=\Prob{\text{Nodes in $S$ have collectively ${k+1 \choose 2}$ long ties to other nodes in $K^{W}_{m,\gamma}(n)$}}$
\end{enumerate}
 \begin{align*}
P_1 &= \sum_{ \forall\ a \in G: d(u, a) \geq n^{1/2-\epsilon}} \frac{\lambda}{(d(u,a))^{\gamma}}
\in O \left (\frac{n}{(n^{\frac{1}{2} - 2\epsilon})^{\gamma}} \right )= O \left( n^{1-(1/2 - \epsilon)\gamma} \right)  \\
P_2  &\leq  \sum^{m(k^2-k+1)}_{i={k+1 \choose 2}} {m(k^2-k+1) \choose i} O\left( (P_{1})^{i} \right) (1-P_{1})^{m(k^2-k+1)-i} \\
 &\leq \sum^{m(k^2-k+1)}_{i={k+1 \choose 2}} {m(k^2-k+1) \choose i} O\left( (P_{1})^{i} \right) \\
 &\in  O \left (P_{1}^{{k+1 \choose 2}} \right ) \\
&=O \left ((n^{1-(1/2 - \epsilon)\gamma})^{{k+1 \choose 2}} \right )
\end{align*}
The outgoing edges of different nodes in $S$ are independent of each other. While the outgoing edges of a single node are not completely independent, they are close enough:  No matter what configuration the first (up to) $j-1$ ties have, the next tie has a $O(P_1)$  probability of being a long range tie.  Thus, the probability that any given vertex has $j$ long-range ties is bounded by $O((P_1)^j)$.

We know that the highest degree of nodes in the graph, $H(n)$, is at most $\log n$ with probability at least $1-n^{-m \log \log n +1}$.

Consider all connected subsets of size $k^2-k+1$ in $B$. If $H(n) \leq \log n$,
then there are at most $n\times \left((k^2 -k) \log n \right)^{k^2-k}$ such connected subsets in the graph, because there are $n$ choices for the first node and at most $(k^2-k) \log n$ choices for each of the next $k^2-k$ nodes.
If $H(n) > \log n$, the number $n^{k^2-k+1}$ gives a trivial bound on the number of all such subsets.
We want to compute the following probability:
\begin{align*}
P_{3} &=\Prob{\text{Any set of $k^2-k+1$ connected nodes has $\geq {k+1 \choose 2}$ long ties connected to its vertices}}
\end{align*}
We condition $P_3$ on $H(n)$ being less/bigger than $\log n$.
\begin{align*}
P_{3} & \leq P_{2}\times n \left ( (k^2 -k)\log n \right)^{k^2-k}  \times (1-n^{-m \log \log n +1} )  + P_{2} \times n^{k^2-k+1} \times n^{-m \log \log n +1} \\
&= O\left( n^{{k+1 \choose 2}(1-(1/2 - \epsilon)\gamma)}\times
  n \log^{k^2-k} n \right) \subseteq O\left( n^{\zeta} \right) 
\end{align*}
We want $P_3$ to be polynomially small, so we require that:
\begin{align*}
  1 + {k+1 \choose 2} \left (1-\frac{\gamma}{2}+\gamma\epsilon \right )  < \zeta<  0 
 \qquad &\Rightarrow  \quad 1+ {k+1 \choose 2} -{k+1 \choose 2} \gamma/2 < 0 \\
& \Rightarrow \quad  \gamma > \frac{2(k^2+k+2)}{k(k+1)}= \alpha_k
 \end{align*} 
But this means that 
this inequality can be satisfied because
\begin{align*}
1 +  {k+1 \choose 2}(1-(1/2 - \epsilon)\gamma) &<   1 +  {k+1 \choose 2}\left (1- \gamma/2 +\frac{\gamma(1-\alpha_k/\gamma)}{2} \right ) \quad 
\text{since  $\ \epsilon< \frac{1-\alpha_k/\gamma}{2}$} \\
                                             &  =  1 +  {k+1 \choose 2} (1-\alpha_k/2)  \\
                                             &  \leq  1 + {k+1 \choose 2} \left(1-\left (1 + \frac{1}{{k+1 \choose 2}} \right) \right)   \quad \text{since  $\ \alpha_k/2 =1 + \frac{1}{{k+1 \choose 2}}$} \\
                                             &  =  0  .
 \end{align*}
\end{proof}

\section{Related Work}\label{sec:related}

\paragraph{General models on complex contagions.} The model of $k$-complex contagions belongs to the general family of \emph{threshold models}, in which each node may have a different threshold on the number of infected edges/neighbors needed to become infected~\cite{Gran78}.  The threshold model is motivated by certain coordination games studied in the economics literature in which a user maximizes its payoff when adopting the behavior as the majority of its neighbors.  Many of the studies focus on the stable states, and structural properties that prevent complete adoption of the advanced technology (better behaviors)~\cite{Morris97contagion}. Montanari and Saberi~\cite{Mon09} analyze a particular dynamics and show that the convergence time is exponential in a quantity called the `tilted cutwidth' of the graph and characterized the convergence time for Kleinberg's small world model.  The main difference in our model, compared to the coordination game, is that we do not have a competing old behavior and all users once infected will remain infected.

\paragraph{Complex contagions on time-evolving graphs.} In our recent work to be reported elsewhere~\cite{Ebrahimi14how}, we analyzed the behavior of complex contagions in time evolved networks, one of which is the famous preferential attachment model~\cite{barabasi99emergence}. We proved that a $k$-complex contagion can spread in $O(\log n)$ number of rounds in such graphs. 

\paragraph{Complex contagions with randomly chosen seeds.}
In \emph{bootstrap percolation}~\cite{Chalupa79bootstrap,adler91bookstrap}, all nodes have the same threshold but initial seeds are randomly chosen. Here, the focus is to examine the threshold of the number of initial seeds with which the infection eventually `percolates', i.e. diffuses to the entire network. Studies have been done on the random Erdos-Renyi graph~\cite{Janson12} , the random regular graphs~\cite{BaloghP07}, and the configuration model~\cite{Amini10}, etc~\cite{Amini12}. All of these results show that for a complex contagion to percolate, the number of initial seeds is a growing function of the network size and in many cases a constant fraction of the entire network. In contrast, we always start with a constant number of seeds and we would like to examine whether a fast spreading is possible.

\bibliographystyle{abbrv}
\bibliography{complex-contagion-Kleinberg}

\begin{thebibliography}{10}

\bibitem{adler91bookstrap}
J.~Adler.
\newblock {Bootstrap percolation}.
\newblock {\em Physica A: Statistical and Theoretical Physics},
  171(3):453--470, Mar. 1991.

\bibitem{Amini10}
H.~Amini.
\newblock Bootstrap percolation and diffusion in random graphs with given
  vertex degrees.
\newblock {\em Electr. J. Comb.}, 17(1), 2010.

\bibitem{Amini12}
H.~Amini and N.~Fountoulakis.
\newblock What {I} tell you three times is true: bootstrap percolation in small
  worlds.
\newblock In {\em Proceedings of the 8th international conference on Internet
  and Network Economics}, pages 462--474, 2012.

\bibitem{BaloghP07}
J.~Balogh and B.~Pittel.
\newblock Bootstrap percolation on the random regular graph.
\newblock {\em Random Struct. Algorithms}, 30:257--286, 2007.

\bibitem{barabasi99emergence}
A.~Barab\'asi and R.~Albert.
\newblock Emergence of scaling in random networks.
\newblock {\em Science}, 286:509--512, 1999.

\bibitem{centola2010spread}
D.~Centola.
\newblock The spread of behavior in an online social network experiment.
\newblock {\em science}, 329(5996):1194, 2010.

\bibitem{G08}
D.~Centola and M.~Macy.
\newblock {Complex Contagions and the Weakness of Long Ties}.
\newblock {\em American Journal of Sociology}, 113(3):702--734, 2007.

\bibitem{Chalupa79bootstrap}
J.~Chalupa, P.~L. Leath, and G.~R. Reich.
\newblock Bootstrap percolation on a bethe lattice.
\newblock {\em Journal of Physics C: Solid State Physics}, 12(1):L31, 1979.

\bibitem{Coleman:1966}
J.~S. Coleman, E.~Katz, and H.~Menzel.
\newblock {\em {Medical Innovation: A Diffusion Study}}.
\newblock Bobbs-Merrill Co, 1966.

\bibitem{Ebrahimi14how}
R.~Ebrahimi, J.~Gao, G.~Ghasemiesfeh, and G.~Schoenebeck.
\newblock How complex contagions spread quickly in the preferential attachment
  model and other time-evolving networks.
\newblock http://arxiv.org/abs/1404.2668, 2014.

\bibitem{Ghasemiesfeh:2013:CCW}
G.~Ghasemiesfeh, R.~Ebrahimi, and J.~Gao.
\newblock Complex contagion and the weakness of long ties in social networks:
  revisited.
\newblock In {\em Proceedings of the fourteenth ACM conference on Electronic
  Commerce}, pages 507--524, 2013.

\bibitem{Gran78}
M.~Granovetter.
\newblock Threshold models of collective behavior.
\newblock {\em The American Journal of Sociology}, 83(6):1420--1443, 1978.

\bibitem{Janson12}
S.~Janson, T.~Luczak, T.~Turova, and T.~Vallier.
\newblock Bootstrap percolation on the random graph ${G}_{n,p}$.
\newblock {\em Annals of Applied Probability}, 22(5):1989--2047, 2012.

\bibitem{kleinberg00small}
J.~Kleinberg.
\newblock The small-world phenomenon: an algorithm perspective.
\newblock In {\em Proceedings of the 32-nd annual ACM symposium on Theory of
  Computing}, pages 163--170, 2000.

\bibitem{Mar04}
C.~Martel and V.~Nguyen.
\newblock Analyzing kleinberg's (and other) small-world models.
\newblock In {\em Proceedings of the twenty-third annual ACM symposium on
  Principles of distributed computing}, pages 179--188, 2004.

\bibitem{Mon09}
A.~Montanari and A.~Saberi.
\newblock Convergence to equilibrium in local interaction games.
\newblock {\em SIGecom Exch.}, 8(1):11:1--11:4, July 2009.

\bibitem{Morris97contagion}
S.~Morris.
\newblock Contagion.
\newblock {\em Review of Economic Studies}, 67:57--78, 2000.

\bibitem{nmw00}
M.~E.~J. Newman, C.~Moore, and D.~J. Watts.
\newblock Mean-field solution of the small-world network model.
\newblock {\em Physics Review Letters}, 84:3201--3204, 2000.

\bibitem{Mar05}
V.~Nguyen and C.~Martel.
\newblock Analyzing and characterizing small-world graphs.
\newblock In {\em Proceedings of the Sixteenth Annual ACM-SIAM Symposium on
  Discrete Algorithms}, SODA '05, pages 311--320. Society for Industrial and
  Applied Mathematics, 2005.

\bibitem{Romero11}
D.~M. Romero, B.~Meeder, and J.~Kleinberg.
\newblock Differences in the mechanics of information diffusion across topics:
  idioms, political hashtags, and complex contagion on twitter.
\newblock In {\em Proceedings of the 20th international conference on World
  wide web}, pages 695--704, 2011.

\bibitem{ugander12}
J.~Ugander, L.~Backstrom, C.~Marlow, and J.~Kleinberg.
\newblock Structural diversity in social contagion.
\newblock {\em Proc. National Academy of Sciences}, 109(16):5962--5966, April
  2012.

\end{thebibliography}

\appendix

%
%


\section{Proofs of Lemma~\ref{lem:recurrence-W-gen-k} and Lemma~\ref{lem:recurrence-I-gen-k}} 
\label{app:upperbound-lemmas}
In this section, we first provide the proof of Lemma~\ref{lem:recurrence-W-gen-k} for general $k$;
and then we prove Lemma~\ref{lem:recurrence-I-gen-k}. Both of these lemmas are computational
tasks and follow from definitions of the  $K^{W}_{m,\gamma}(n)$ and  $K^{I}_{m,\gamma}(n)$ models
and Definition~\ref{def:rec-spreading}. As before $\alpha_k =  \frac{2(k^2 + k + 2)}{k(k + 1)}$ and $\beta_k = \frac{2(k+1)}{k}$.

\begin{restate} {of Lemma~\ref{lem:recurrence-W-gen-k}}
Fix constants $2 < \gamma< \alpha_k$, $0 < \delta < 1- \alpha_k \gamma$,  $\frac{1-\delta}{\left({k + 1 \choose 2} + 1 \right)(1 - \delta) + {k + 1 \choose 2}\gamma/2} \leq c$, and $0 < d$, then
 $K^{W}_{m,\gamma}(n)$  is  $(\delta, c, d, k)$-recursively spreading.  
 \end{restate}

\begin{proof}
Let $r = \left(\frac{d 2^k\prod_{i = 1}^k i^{k-i+ 1}}{ \lambda^{{k + 1 \choose 2}}}\right)^c$
Fix  $\ell > (r \log ^c n)^{(\frac{1}{1-\delta})}$, and $\ell$-sized square $S$ and two disjoint $\ell^{(1-\delta)}$-sized subsquares $A$ and $B$, and partition $B$ into $\ell^{(1-\delta)}$/k disjoint pairs of (grid) neighbor nodes  $(u_1, \ldots, u_k)$. Assume that $a \in A$. Let $d(u,a)$ be the Manhattan distance of $u$ and $a$.
 \begin{enumerate}[i)] \denselist
 \item $P_{1}(\ell)=\Prob{\text{$u$ has a weak tie to $a \in A$, via a particular edge}}$
 \item $P_{2}(\ell)=\Prob{\text{$u$ has a weak tie to $A$, via a particular edge}} \geq |A| \times P_{1}(\ell)$
 \item $Q_{s}(\ell)=\Prob{\text{$u$ is connected to $A$, via $s$ distinct weak ties}} $
 \item $P_{4}(\ell)=\Prob{\text{$(u_1, \ldots, u_k)$ form a \emph{new seed} in $B$}} $
  \item$P_{5}(\ell)=\Prob{\text{a new seed forms in $B$}} $
 \end{enumerate}
 \begin{align*}
P_{1}(\ell) & \geq  \frac{\lambda}{(d(u,a))^{\gamma}}  =   \frac{\lambda}{ \sqrt{\ell^{\gamma}} } ;\\
  P_{2}(\ell) &\geq |A| \times P_{1}(\ell)  = \lambda \ell^{1-\delta-\gamma/2} ;\\
 Q_s(\ell) & \geq |\{ (a_1,\ldots, a_s) | a_1, \ldots, a_s \in A \}| \times \Prob{\text{$u$ has a weak tie to $a_1, \ldots, a_s$}}  \\
 &\geq  {|A| \choose s} P_{1}(\ell)^{s}  \geq \frac{ \lambda^s \ell^{s(1-\delta-\gamma/2)}}{2 s!};\\
 P_{4}(\ell) & \geq Q_{k}(\ell)\times Q_{k-1}(\ell) \times \cdots \times Q_{1}(\ell)  = \frac{ \lambda^{{k + 1 \choose 2}} \ell^{{k + 1 \choose 2}(1 -\delta- \gamma/2)} }{2^k\prod_{i = 1}^k i^{k  -i + 1}} ;\\
P_{5}(\ell)  &\geq 1 - (1 - P_{4}(\ell))^{|B|/2} \geq  1 - \exp\{P_{4}(\ell) |B|/2\} \\
&= 1 - \exp\{-\frac{ \lambda^{{k + 1 \choose 2}}}{2^k\prod_{i = 1}^k i^{k-i+ 1}} \ell^{\left ({k + 1 \choose 2} + 1 \right )(1 -\delta) - {k + 1 \choose 2}\gamma/2} \}
\end{align*}
We want $P_5(\ell)$ to be very close to $1$ and that means that we want the power of $\ell >0$, which
gives us:
$$
 {\left ({k + 1 \choose 2} + 1 \right )(1 -\delta) - {k + 1 \choose 2}\gamma/2} > 0  \ \Rightarrow \ \gamma <  \alpha_k,  \ \delta< 1 - \frac{ \gamma}{\alpha_k}.
$$
Also, note that $P_5$ is increasing in $\ell$. Therefore, the smallest probability happens when
then size of $\ell$ is the smallest; that is $\ell^{1-\delta} =r \log^c n$ and $\ell = (r \log^c n)^{1/(1-\delta)}$.
We want $1-P_{5}\left((r \log^c n)^{1/(1-\delta)} \right)$ to be polynomially small even when $\ell$
takes on its smallest value.
This requires that the power of $\log n$ be $ \geq 1$:
\begin{align*}
& \frac{c}{1-\delta} \left ({ \left ({k + 1 \choose 2} + 1 \right )(1 -\delta) - {k + 1 \choose 2}\gamma/2}  \right ) \geq 1  \\
& \Rightarrow \ c \geq  \frac{1- \delta}{\left ({k + 1 \choose 2} + 1 \right )(1 -\delta) - {k + 1 \choose 2}\gamma/2};
\end{align*}
and by replacing $r= \left(\frac{d 2^k\prod_{i = 1}^k i^{k-i+ 1}}{ \lambda^{{k + 1 \choose 2}}}\right)^c$, we get that
\begin{align*}
P_{5}(\ell)  & \geq 1 -\exp\{-\frac{ \lambda^{{k + 1 \choose 2}}}{2^k\prod_{i = 1}^k i^{k-i+ 1}} ((r \log^c n)^{\frac{1}{(1-\delta)}})^{\left ({k + 1 \choose 2} + 1 \right )(1 -\delta) - {k + 1 \choose 2}\gamma/2} \}  \\
&\geq 1-\exp\{-d \log n\} \geq 1-n^{-d}.
\end{align*}
Finally, notice that this probability is much stronger than the statement of Lemma.
\end{proof}


%


\begin{restate}{of Lemma~\ref{lem:recurrence-I-gen-k}} 
Fix constants $2 < \gamma< \beta_k$, $0 <  \delta < 1- \frac{\gamma}{\beta_k}$,  $\frac{1- \delta}{(k + 1)(1 -\delta) - k \gamma/2} \leq c$, and $0 < d$, then
 $K^{I}_{m,\gamma}(n)$  is  $(\delta, c, d, k)$-recursively spreading. 
 \end{restate}

\begin{proof}
Let $r= \left(\frac{d k^{(k^2\gamma + 1)} }{ \lambda^{k^2}}\right)^c$
Fix  $\ell > (r \log ^c n)^{(\frac{1}{1-\delta})}$, an $\ell$-sized square $S$ and two disjoint $\ell^{(1-\delta)}$-sized subsquares $A$ and $B$, and partition $B$ into $\ell^{(1-\delta)}/k$ disjoint pairs of (grid) neighbor nodes  $(u_1, \ldots, u_k)$. Assume that $a \in A$. Let $d(u,a)$ be the Manhattan distance of $u$ and $a$.
 \begin{enumerate}[i)] \denselist
 \item $P_{1}(\ell)=\Prob{\text{$u$ has a weak tie to $a \in A$, via a particular edge}}$
 \item $P_{2}(\ell)=\Prob{\text{$u$ has a weak tie to $A$, via a particular edge}} \geq |A| \times P_{1}(\ell)$
 \item $Q_{1}(\ell)=\Prob{\text{$u_1$ is connected to $A$, via $k$ distinct weak ties}} $
 \item $Q'_{s}(\ell)=\Prob{\text{$u_s$ has $k$ ties to $u_1$}}$
  \item $P_4(\ell)=\Prob{\text{$(u_1, \ldots u_k)$ form a \emph{new seed} in $B$}} $
  \item$P_{5}(\ell)=\Prob{\text{a new seed forms in $B$}} $
 \end{enumerate}
 \begin{align*}
P_{1}(\ell) & \geq  \frac{\lambda}{(d(u,a))^{\gamma}}  =   \frac{\lambda}{ \sqrt{\ell^{\gamma}} } ;\\
  P_{2}(\ell) &\geq |A| \times P_{1}(\ell)  = \lambda \ell^{1-\delta-\gamma/2} ;\\
 Q_1(\ell) & \geq P_{2}^k \geq  \lambda^k \ell^{k(1-\delta-\gamma/2)};\\
 Q'_s(\ell) & \geq \left(\frac{\lambda}{k^{\gamma}} \right)^k;\\
 P_{4}(\ell) & \geq Q_{1}(\ell)\times Q'_{2}(\ell) \times \cdots \times Q'_{k}(\ell)  = \frac{ \lambda^{(k^2)} \ell^{k(1 -\delta- \gamma/2)} }{k^{(\gamma k^2)}};\\
P_{5}(\ell)  &\geq 1 - (1 - P_{4}(\ell))^{|B|/2} \geq  1 - \exp\{P_{4}(\ell) |B|/k\}  \\
&= 1 - \exp\{-\frac{ \lambda^{(k^2)}}{k^{(\gamma k^2 +1)}} \ell^{(k + 1)(1 -\delta) - k \gamma/2} \}
\end{align*}
We want $P_5(\ell)$ to be very close to $1$ and that means that we want the power of $\ell >0$, which
gives us:
$$
 {(k + 1)(1 -\delta) - k \gamma/2} > 0  \ \Rightarrow \ \gamma <  \beta_k,  \ \delta< 1 - \frac{ \gamma}{\beta_k}.
$$
Also, note that $P_5$ is increasing in $\ell$. Therefore, the smallest probability happens when
then size of $\ell$ is the smallest; that is $\ell^{1-\delta} =r \log^c n$ and $\ell = (r \log^c n)^{1/(1-\delta)}$.
We want $1-P_{5}\left((r \log^c n)^{1/(1-\delta)} \right)$ to be polynomially small even when $\ell$
takes on its smallest value.
This requires that the power of $\log n$ be $ \geq 1$:
$$
\frac{c}{1-\delta}((k + 1)(1 -\delta) - k \gamma/2) \geq 1  \  \Rightarrow \ c \geq  \frac{1- \delta}{(k + 1)(1 -\delta) - k \gamma/2};
$$
and by replacing $r= \left(\frac{d k^{(k^2\gamma + 1)} }{ \lambda^{k^2}}\right)^c$, we get that
\begin{align*}
P_{5}(\ell)  &\geq 1 -  \exp\{-\frac{ \lambda^{(k^2)}}{k^{(\gamma k^2 ++1}}\left(  (r \log^c n)^{1/(1-\delta)} \right)^{(k + 1)(1 -\delta) - k \gamma/2} \\
&\geq 1-\exp\{-d \log n\} \geq 1-n^{-d}.
\end{align*}
Finally, notice that this probability is much stronger than the statement of Lemma.
\end{proof}

\end{document}